\documentclass[10pt,journal,twocolumn,twoside]{IEEEtran}
\usepackage{balance}
\usepackage{threeparttable}
\usepackage{longtable}
\usepackage{multirow}
\usepackage{cite}
\usepackage{color}
\usepackage{multirow}
\usepackage{graphicx}
\usepackage{epstopdf}
\usepackage{mathtools}
\usepackage{multirow}
\usepackage{algorithmic}
\usepackage{algorithm}
\usepackage{array}
\usepackage{epsfig}
\usepackage{epstopdf}
\usepackage{dblfloatfix}
\usepackage{url}
\usepackage{cases}
\usepackage{threeparttable}
\usepackage{booktabs}
\usepackage{tabularx}
\usepackage{diagbox}
\usepackage{mathrsfs}
\usepackage{amsthm,amsmath,amssymb}
\usepackage{multirow}

\ifCLASSOPTIONcompsoc
  \usepackage[caption=false,font=normalsize,labelfont=sf,textfont=sf]{subfig}
\else
  \usepackage[caption=false,font=footnotesize]{subfig}
\fi

\usepackage{url}
\begin{document}

\title{An Efficient Construction Method Based on Partial Distance of Polar Codes with Reed-Solomon Kernel}

\author{Jianhan Zhao, Wei Zhang, Yanyan Liu

\thanks{
Jianhan Zhao, Wei Zhang are with the School of Microelectronics, Tianjin University, Tianjin 300072, China (e-mail: zjh1997@tju.edu.cn; tjuzhangwei@tju.edu.cn).

Yanyan Liu is with the Tianjin Key Laboratory of Photoelectronic Thin Film
Devices and Technology, Nankai University, Tianjin 300071, China (e-mail:
lyytianjin@nankai.edu.cn).

Digital Object Identifier -----}
}

\markboth{	,~Vol.~0, No.~0, month~2023}%
{ {\it \MakeLowercase .} An Efficient Construction Method Based on Partial Distance of Polar Codes with Reed-Solomon Kernel}

\maketitle

\begin{abstract} Polar codes with Reed-Solomon (RS) kernel have great potential in next-generation communication systems due to their high polarization rate. In this paper, we study  the polarization characteristics of RS polar codes and propose two types of partial orders (POs) for the synthesized channels, which are supported by validity proofs. By combining these partial orders, a Partial Distance-based Polarization Weight (PDPW) construction method is presented. The proposed method achieves comparable performance to Monte-Carlo simulations while requiring lower complexity.  Additionally, a Minimum Polarization Weight Puncturing (MPWP) scheme for rate-matching is proposed to enhance its practical applicability in communication systems. Simulation results demonstrate that the RS polar codes based on the proposed PDPW construction outperform the 3rd Generation Partnership Project (3GPP) NR polar codes in terms of standard code performance and rate-matching performance.

\end{abstract}


\IEEEpeerreviewmaketitle

\section{Introduction}\label{sec:intro}

\IEEEPARstart{P}{olar} codes presented by Arikan \cite{b1-1} were proved to achieve the capacity of binary-input symmetric discrete memoryless channels (BI-DMC). In a  polar code with block length $N$ and  rate of $R$, the key process is to select $NR$ sub-channels with high reliability to directly transmit information, thereby approximating the channel capacity. To decode polar codes, successive cancellation (SC) decoders were initially used, which have a time complexity of  $\mathcal{O}(NlogN)$. 

However, as the incomplete channel polarization in the finite regime, the performance of polar codes with short or moderate block lengths is not ideal. Therefore, researchers have focused on improving polar codes performance by developing decoding algorithms and construction methods. In terms of decoding algorithms,  the successive cancellation list (SCL) \cite{b1-2,b1-3,b1-4} decoding and the successive cancellation stack (SCS) \cite{b1-5,b1-6,b1-7} decoding have been proposed and improved, which achieve a good decoding performance.  Concatenating with a Cyclic Redundancy Check (CRC) code further improves the performance. In terms of construction methods, several algorithms have been proposed, such as Tal-Vardy method  \cite{b1-8} and Gaussian approximation (GA) \cite{b1-9}. Practical communication systems such as the 5G system require channel coding with flexible codewords to take full advantage of limited resources and accommodate changeable code rates. To address this, Quasi-uniform puncturing (QUP) \cite{b1-10} and Information Set Approximation Puncturing (ISAP) \cite{b1-11} have been proposed, demonstrating good rate-matching performance.  Due to these excellent performance characteristics, polar codes  have been accepted as the coding scheme for the 5G new radio (5G NR)  control channel \cite{b1-12}.

Multi-kernel polar codes have recently gained attention in \cite{b1-13,b1-14,b1-15,b1-16,b1-17,b1-18} due to their superior performance. Specifically, polar codes over $\mathcal{GF}(q)$ with  Reed-Solomon (RS) kernel can achieve a higher polarization rate than the Arikan kernel\cite{b1-16,b1-19}.  Several pioneering works have been conducted to design polar codes with RS kernel.  Results in \cite{b1-16} demonstrate  that the polar codes with  4-dimension RS kernel could already achieve a better performance than most binary kernels. Cheng $et$ $al$. have proposed a look-up table-based encoding algorithm \cite{b1-20} for polar codes with  RS kernel, which significantly reduces the encoding complexity. In \cite{b1-21}, an effective decoding scheme using "dynamic frozen symbols" is proposed to improve decoding performance. Furthermore,  \cite{b1-22} presents a piecewise sequence rate-matching scheme to achieve flexible code rates.

However,  few studies have investigated the construction of polar codes with RS kernel. Most previous studies, as mentioned above, have relied on Monte-Carlo simulations, which are not efficient and are difficult to apply to practical communication systems. In  recent studies, partial orders (POs) were shown to exist between the reliabilities of the sub-channels of polar codes with Arikan kernel \cite{b1-23,b1-24,b1-25}. Based on these partial orders, construction methods with low complexity were proposed in \cite{b1-26,b1-27}.  These works enlighten us to pursue an efficient construction scheme of polar codes with RS kernel for practical transmission channels. 

In this paper, we focus on studying the polarization characteristics of polar codes with RS kernel, especially the relationship between the sub-channel index and its channel reliability. We first present partial orders for polar codes with RS kernel that are independent of the underlying channel $W$. These POs are based on the observation of the relationship between the reliability of symbol sub-channels  and the q-ary representation of their channel indexes. Combing the POs results, we propose a  Partial Distance-based Polarization Weight (PDPW) construction method with lower complexity, while ensuring performance guarantees. Finally, based on the polarization weight, an efficient rate-matching scheme of polar codes with RS kernel is proposed which can  be applied to more flexible communication systems. 

The remainder of the paper is organized as follows. In Section \ref{sec:section2} we provide a brief overview of polar codes and the Reed-Solomon (RS) kernel. Section \ref{sec:section3} proposes two partial orders for the polarization properties of polar codes with RS kernel and describes the proposed PDPW construction method.  In Section \ref{sec:section4}, a Minimum Polarization Weight Puncturing (MPWP) scheme is presented to evaluate the rate-matching performance. Then, in Section \ref{sec:section5}, the simulation results  of the proposed algorithms are presented. Finally, Section \ref{sec:section6} concludes this work.

\section{PRELIMINARIES}\label{sec:section2}

\newtheorem{Definition}{Definition}
\newtheorem{Example}{Example}

\subsection{Notations and Channel Parameters}
Throughout the paper, $n$, $t$ and $m$ is non-negative integers. Let  $N$ and $N_{b}$ denote the symbol length and bit length of a code, respectively. Denote  $u_{0}^{N-1}$  as the row vector $\left\{u_{0},u_{1},...,u_{N-1}\right\}$.

\begin{Definition} \label{Definition2-1}
Let $\mathcal{X}$ and $\mathcal{Y}$ be the $q$-ary input alphabet and output  alphabet for a channel  $W:\mathcal{X} \to \mathcal{Y}$. We define the mutual information of channel $W$  as follows:

\begin{equation}
I(W)=\sum_{x \in \mathcal{X}} \sum_{y \in \mathcal{Y}} \frac{1}{q} W(y \mid x) \log \frac{W(y \mid x)}{\frac{1}{q} \sum_{x^{\prime} \in \mathcal{X}} W\left(y \mid x^{\prime}\right)}
\end{equation}

\end{Definition}

\begin{Definition} \label{Definition2-2}
The Bhattacharyya parameter of a $q$-ary channel $W$ is defined
as
\begin{equation}
Z(W)=\frac{1}{q(q-1)} \sum_{\substack{x \in \mathcal{X}, x^{\prime} \in \mathcal{X}, x \neq x^{\prime}}} Z_{x, x^{\prime}}(W)
\end{equation}
where for any input symbols $x,x^{\prime} \in \mathcal{X}$ the Bhattacharyya parameter is defined as 
\begin{equation}
Z_{x, x^{\prime}}(W)=\sum_{y \in \mathcal{Y}} \sqrt{W(y \mid x) W\left(y \mid x^{\prime}\right)}
\end{equation}

\end{Definition}

\begin{Definition} \label{Definition2-3}
Let $\mathcal{X}$ and $\mathcal{Y}$ be the $q$-ary input alphabet and output  alphabet for a channel  $W:\mathcal{X} \to \mathcal{Y}$ with the basic input $x\in \mathcal{X}$ and $y\in \mathcal{Y}$. Denote the $\mathcal{D}_x=\{y \in \mathcal{Y} \mid W(y \mid x)>W(y \mid$ $\left.\left.x^{\prime}\right), \forall x^{\prime} \in \mathcal{X}, x^{\prime} \neq x\right\}$. The  average error probability of the maximum-likelihood estimation of channel $W$ is defined as :

\begin{equation}
P_e(W)=\frac{1}{q} \sum_{x \in \mathcal{X}} \sum_{y \in \mathcal{D}_x^c} W(y \mid x) .
\end{equation}

where $P_e(W)$  is bounded as
\begin{equation}
P_e(W)\leq(q-1)Z(W)
\end{equation}

\end{Definition}

\begin{Definition} \label{Definition2-4}
(Stochastic Degradation \cite{b1-8}): Let $W_{1},W_{2},W_{3}$ be the  channels $\mathcal{X} \rightarrow \mathcal{Z}$ , $\mathcal{X} \rightarrow \mathcal{Y}$ and  $\mathcal{Y} \rightarrow \mathcal{Z}$ , we say $W_{1}$ is stochastically degraded with respect to $W_{2}$, denoted as $W_{2} \succeq W_{1}$ if their  relationship exist such that 
\begin{equation}
W_{1}(z \mid x)=\sum_{y \in \mathcal{Y}} W_{2}(y \mid x) W_{3}(z \mid y)
\end{equation}

\end{Definition}

\begin{Definition} \label{Definition2-5}
(Reliability measure): For two synthetic channels  $W_{N}^{(i)}$ and  $W_{N}^{(j)}$,
if $W_{N}^{(i)}\succeq W_{N}^{(j)}$ or $i\succeq j$, we say  $W_{N}^{(i)}$ is more reliable than $W_{N}^{(j)}$ then 
\begin{equation}
I( W_{N}^{(i)}) \succeq I( W_{N}^{(j)})
\end{equation}

\begin{equation}
Z( W_{N}^{(i)})  \preceq Z( W_{N}^{(j)})
\end{equation}

\begin{equation}
P_e( W_{N}^{(i)}) \preceq P_e( W_{N}^{(j)})
\end{equation}

\end{Definition}

\subsection{Arikan Polar Codes }
For an Arakan polar code with length $N=2^{n}$, the indices of the sub-channels $W_{N}^{(i)}$ after channel polarization can be expressed as  $\left\{0,1,...,N-1\right\}$.  Let $\mathcal{I}\subseteq \left\{0,1,...,N-1\right\}$ be the information-bit set that contains the indices of the $K$ most reliable sub-channels carrying information bits. Its complement $\mathcal{F}$ represents the frozen-bit set. The generator matrix is $F_{N}=G_{2}^{\otimes n}$, where kernel $G_{2}\triangleq \begin{bmatrix} 1 & 0 \\ 1 & 1 \end{bmatrix}$ and $\otimes$ denotes the  Kronecker product of the matrix with itself. The block error performance of polar codes can be expressed as 
\begin{equation}
P_B=o\left(2^{-N^\epsilon}\right), \quad \forall \epsilon<E(G)
\end{equation}
where the threshold $E(G)$ of $\epsilon$ is the exponent of kernel $G$ as introduced in \cite{b1-15}. For instance, the Arikan kernel $G_2$ has an exponent of 0.5, and  kernels with larger exponents can achieve better performance.

\subsection{Reed-Solomon Kernels  }
 Reed-Solomon Kernels can be regarded as a generalization of Arıkan kernel. In the case of a channel with input alphabet $\mathcal{GF}(q=2^t)$, one can  obtain the $q$-ary matrix $G_{q}$. The submatrix of $G_{q}$ that consists of the $i$-th row to the $(q-1)$-th row is a generator matrix of a generalized Reed-Solomon code. Therefore, we refer to the $q$-ary matrix $G_{q}$ as the Reed-Solomon kernel,e.g.
\begin{equation}
G_{q}=\left(\begin{array}{cccccc}
1 & 1 & \cdots & 1 & 1 & 0 \\
\alpha_{q-2}^{q-2} & \alpha_{q-3}^{q-2} & \cdots  & \alpha_1^{q-2}&1 & 0 \\
\vdots & \vdots & \ddots & \vdots & \vdots & \vdots \\
\alpha_{q-2}^1 & \alpha_{q-3}^1 & \cdots  & \alpha_1^1 & 1 & 0 \\
\alpha_{q-2}^0 & \alpha_{q-3}^0 & \cdots & \alpha_1^0 & 1 & \gamma
\end{array}\right)
\end{equation}
where  $\alpha_{i}$ are primitive element of $\mathcal{GF}(q)$ and $\gamma$ is a  non-zero element of  $\mathcal{GF}(q)$.

By analogy with the Arikan kernel, we can construct  an RS polar code of symbol length $N=q^{m},m \in \mathbb{Z}^{+}$  over $\mathcal{GF}(q)$. For a $q$-ary symbols information  sequence $s_{0}^{N-1}$, its corresponding encoded codeword can be expressed as  $c_{0}^{N-1}=s_{0}^{N-1}F_N=s_{0}^{N-1}G_{q}^{\otimes m}$. 
If symbol $c_{i}$ are transmitted over a memoryless output-symmetric channel $W(y|c)$, then after channel polarization, the $q$-ary  synthetic sub-channels with their respective transition probabilities can be defined as follow:

\begin{equation}
W_{N}^{(i)}\left(y_0^{N-1}, s_0^{i-1} \mid s_i\right)=\frac{1}{q^{N-1}} \sum_{s_{i+1}^{N-1}} \prod_{i=0}^{N-1} W\left(y_{i} \mid\left(s_{0}^{N-1} F_N\right)_i\right) \label{ei}
\end{equation}
More specifically, let $(i_{m-1},...,i_1,i_0)$ be the $q$-ary representation of the index $i$,i.e.,
\begin{equation}
i=\sum_{k=0}^{m-1} i_k q^{k} \label{e6}
\end{equation}
Then, we define the $W_N^{(i)}$ as 
\begin{equation}
W_N^{(i)}=\left(\left(\left(W^{i_0}\right)^{i_1}\right)^{\cdots}\right)^{i_{m-1}}
\end{equation}

In addition, for polar codes with multi-kernel, the kernel exponent is represented as 
\begin{equation}
E(G)= ln(q!)/(q lnq)
\end{equation}
For instance, when $q=2^{2}$, the exponent of the Reed Solomon kernel $G_{4}$ is $ln 24/(4 ln 4) \approx 0.57312$, which is larger than the exponents of most binary linear kernels   in \cite{b1-15}, despite its small and simple structure. Therefore, the RS kernel achieves a higher polarization rate, and its performance is better when constructing codewords of the same block length.

\begin{figure}[!h]
\centering
\includegraphics[width=3.2in]{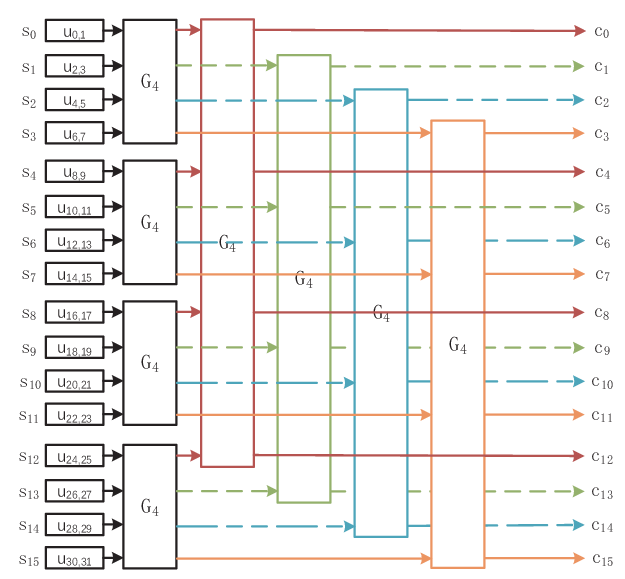}
\caption{  Polar codes with RS kernel with $q=4$ and $m=2$.}
\label{fig:figure1}
\end{figure}

In practice, taking $G_{4}$ kernel as an example, 
\begin{equation}
G_{4}=\left[\begin{array}{cccc}
1 & 1 & 1 & 0 \\
\alpha & \alpha^2 & 1 & 0 \\
\alpha^2 & \alpha & 1 & 0 \\
1 & 1 & 1 & \alpha
\end{array}\right]
\end{equation}
the scheme of RS polar codes with $q=4$ and $m=2$ using the $G_4$ kernel is illustrated in Fig. \ref{fig:figure1}.

When constructing polar codes using the $G_4$ kernel, every two  bits are mapped into a $q-$ary symbol over  $\mathcal{GF}(q)$ as shown below:
\begin{equation}
(0,0) \leftrightarrow 0, \quad(0,1) \leftrightarrow \alpha, \quad(1,0) \leftrightarrow \alpha^2, \quad(1,1) \leftrightarrow \alpha^3.
\end{equation}
Note that the mapping is not unique and the addition and multiplication  are based on the $\mathcal{GF}(q)$ operations in  \cite{b1-20}.

\subsection{Rate-matching Scheme}
In practical  communication system, flexible code length and rate are required. Therefore, the standard codeword needs to be adjusted utlizing rate-matching method before transmission to the channel. A typical rate-matching scheme for RS polar codes  is  illustrated in  Fig.\ref{fig:figure2}. As mentioned  in the  previous subsection, the $N$ size  RS polar codes contain fixed $q^m$ symbols ($N_{b}=t\times q^m$ bits). After a specific rate-matching algorithm, the actually transmitted code block length in bits is $M_b$, where  $M_b<N_{b}$. Thus, the actual transmission block rate is
defined as $R = K_{b}/M_b$.
\begin{figure}[!h]
\centering
\includegraphics[width=3.2in]{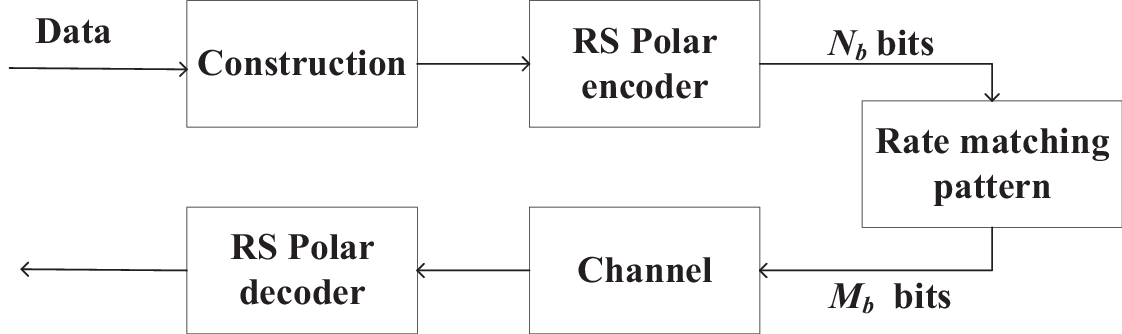}
\caption{ Rate-matching scheme for polar codes with RS kernel}
\label{fig:figure2}
\end{figure}
On the receiver side, the decoding structure remains the same. Since the $N_{b}-M_b$ bits are not sent over the channel, their initial log-likelihood ratios (LLR) for the decoder are set to zero or infinity depending on different rate-matching algorithms.

\section{An Efficient Construction of  POLAR CODES WITH
REED-SOLOMON KERNEL }\label{sec:section3}
In this section, we first investigate  the reliability relationship of sub-channels in a single-layer RS kernel. We then introduce the concept of partial orders (POs) for RS polar codes. Finally, we utilize these POs to propose an efficient construction method for polar codes with RS kernel.

\subsection{ Channel Degradation of Polar Codes with RS Kernel }

Denote the subchannel $W_{G_{q}}^{(i)}$ for a $q \times q$  RS kernel $G_{q}$ as 
\begin{equation}
W_{G_{q}}^{(i)} \triangleq  s_i \rightarrow\left(y_0^{q-1}, s_0^{i-1}\right)
\end{equation}
and the transition probabilities as 
\begin{equation}
W_{G_{q}}^{(i)}\left(y_0^{q-1}, s_0^{i-1} \mid s_i\right)=\frac{1}{q^{q-1}} \sum_{s_{i+1}^{q-1}} W_{G_{q}}\left(y_0^{q-1} \mid s_0^{q-1}\right) \label{e1}
\end{equation}
where 
\begin{equation}
W_{G_{q}}\left(y_0^{q-1} \mid s_0^{q-1}\right) \triangleq \prod_{i=0}^{q-1} W\left(y_i \mid c_i\right)=\prod_{i=0}^{q-1} W\left(y_i \mid\left(s_0^{q-1} G_{q}\right)_i\right)\label{e2}
\end{equation}
As the RS kernel $G_{q}$ is invertible, the following equation holds
\begin{equation}
\sum_{i=0}^{q-1} I\left(W_{G_{q}}^{(i)}\right)=q I(W)
\end{equation}

Given a $q$-ary DMC $W:\left\{0,1,...,q-1\right\} \to \mathcal{Y}$, denote the   transition probabilities of channel $\tilde{W}:\left\{0,...,q-1\right\} \to \mathcal{Y}\times \left\{0,...,q-1\right\}$,

\begin{equation}
\tilde{W}(y, r \mid c)=\frac{1}{q} W(y \mid s+r)
\end{equation}

Then, let $(W \circledast P)$   denote the channel with transition probabilities
$(W \circledast P)\left(y_1, y_2 \mid c\right)=W\left(y_1 \mid c\right) P\left(y_2 \mid c\right)$ and the $W^{\circledast k}$ denote $k$  operations on $W$ itself.
\begin{equation}
W^{\circledast k}\left(y_1^k \mid c\right)=\prod_{j=1}^k W\left(y_j \mid c\right)
\end{equation}
In addition,  it is easy to prove that 
\begin{equation}
(W^{\circledast k})^{\circledast h}=W^{\circledast kh}=\prod_{j=1}^{kh} W\left(y_j \mid c\right)
\end{equation}

Since the RS kernel $G_{q}$ is invertible and polarizing \cite{b1-16},  for all $q$-ary DMCs $W$, there exists a
 channel $W_{G_{q}}^{(i)}$ which is statistically equivalent to either $\tilde{W}^{\circledast k}$ or $W\circledast\tilde{W}^{\circledast k-1}$, where $i\in \left\{0,1,...,q-1\right\}$ and $k\geq2$.  These equivalences can be denoted as $W_{G_{q}}^{(i)} \equiv \tilde{W}^{\circledast k}$ and $W_{G_{q}}^{(i)} \equiv W\circledast\tilde{W}^{\circledast k-1}$, where $\equiv$ means that both channels are stochastically degraded with respect to each other.

Now, we  introduce the concept of partial distance:
\begin{Definition} \label{Definition3-1}
For a kernel $g: \mathcal{X}^q \to \mathcal{X}^q $, the  partial distance $D_i$ is given as 
 \begin{equation}
\begin{aligned}
D_i =\min _{v_{i+1}^{q-1}, w_{i+1}^{q-1}} d\left(g\left(s_0^{i-1}, x, v_{i+1}^{q-1}\right), g\left(s_0^{i-1}, x^{\prime}, w_{i+1}^{q-1}\right)\right)
\end{aligned}
\end{equation}
where $a,b \in \mathcal{X}^q$, $d(a,b)$ is the Hamming distance between $a$ and $b$.
\end{Definition}

Thus, we can define   the RS kernel as $G_q=[g_0^T,...,g_{q-1}^T]^T$.
Since it is  generated by a Reed-Solomon code, which is a  maximum distance separable code, the   partial distance of RS kernel  is $D_i=i+1$.
Then, we can rewrite $W_{G_{q}}^{(i)}$ using equation (\ref{e1}) and (\ref{e2}) as follows: 

\begin{equation}
W_{G_{q}}^{(i)}=\frac{1}{q^{q-1-i}} \sum_{c_{0}^{q-1}\in\Lambda(s_0^{i})}  \prod_{k=0}^{q-1} W\left(y_k \mid c_k\right)
\end{equation}
where $\Lambda(s_0^{i})\subset\left\{0,...,q-1\right\}^{q-1}$ and $c_{0}^{q-1}$ satisfying
\begin{equation}
c_0^{q-1}=\sum_{j=0}^{i-1} s_j g_j+s_i g_i+\sum_{j=i+1}^{q-1} v_j g_j \label{e3}
\end{equation}
for some $v_{i+1}^{q-1}\in\left\{0,...,q-1\right\}^{q-1-i}$. Let $g_\delta=\sum_{j=i+1}^{q-1} \delta_j g_j$ be a codeword and satisfying 
 \begin{equation}
\min _{v_{i+1}^{q-1}, w_{i+1}^{q-1}} d\left(g_i\left(s_0^{i-1}, x, v_{i+1}^{q-1}\right), g_\delta\left(s_0^{i-1}, x^{\prime}, w_{i+1}^{q-1}\right)\right)=D_i
\end{equation}
Due to the linearity of $\left\langle g_{i+1}, \ldots, g_{q-1}\right\rangle$, it's equivalent to say that $c_{0}^{q-1}\in\Lambda(s_0^{i})$ if and only if 
\begin{equation}
c_0^{q-1}=\sum_{j=0}^{i-1} s_j g_j+s_i (g_i+g_\delta )+\sum_{j=i+1}^{q-1} v_j g_j \label{e4}
\end{equation}

Therefore, we can define $G_q^{\prime}=\left[g_0^T, \ldots, g_i^{\prime T}, \ldots, g_{q-1}^T\right]^T$, where $g'_i=g_i+g_\delta$. At this point, the partial distance of $g'_i$ in $G^{\prime}_q$ is equal to $D_i$. Furthermore, the channels $W_{G_{q}}^{(i)}$ and $W_{G'_{q}}^{(i)}$ are equivalent based on (\ref{e3}) and (\ref{e4}).

Now, consider a channel $W_{g,q}^{(i)}$ where a genie provides extra information $(s^{q-1}_{i+1})$ to the decoder of $W_{G'_{q}}^{(i)}$,  then the  $W_{G'_{q}}^{(i)}$ can be seen as a degradation of the channel $W_{g,q}^{(i)}$. The transition probabilities of this $W_{g,q}^{(i)}$ are defined as 

\begin{equation}
W_{g,q}^{(i)}\left(y_0^{q-1}, s_0^{i-1}, s_{i+1}^{q-1}\mid s_i\right)=\frac{1}{q^{q-1}}  \prod_{j=0}^{q-1} W\left(y_i \mid\left(s_0^{q-1} G'_{q}\right)_j\right)
\end{equation}
Since the matrix $G_q^{\prime}$ is invertible, referring to the definition of $D_i$, it follows that there exists a set $J= \left\{j\vert j\in \left\{0,1,...,q-1\right\}\right\}$ and  $\left| J \right|=D_i$  such that the genie-aided channel $W_{g,q}^{(i)}$ can be expressed as below:

\begin{equation}
\begin{aligned}
W_{g,q}^{(i)} & \left(y_0^{q-1}, s_0^{i-1}, s_{i+1}^{q-1} \mid s_i\right) \\
= & \left(\frac{1}{q^{|J|}} \prod_{j \in J} W\left(y_j \mid s_i G'_{q(i,j)}+(s G_q^{\prime})_j-s_i G'_{q(i,j)}\right)\right) \\
& \cdot\left(\frac{1}{q^{q-|J|-1}} \prod_{j \in J^c} W\left(y_j \mid(s G_q^{\prime})_j\right)\right)
\end{aligned}
\end{equation}

It satisfies that  the second term on the right-hand side of the above equality is independent of the input $s_i$. Thus, the $W_{g,q}^{(i)}$ is 
equivalent to either $\tilde{W}^{\circledast D_i}$ or $W\circledast\tilde{W}^{\circledast D_i-1}$ and $Z\left(\tilde{W}^{\circledast D_i}\right)=Z\left(W \circledast \tilde{W}^{\circledast D_i-1}\right)=Z(W)^{D_i}$.
As  the   partial distance of RS kernel  is $D_i=i+1$ and $0<Z(W)<1$, the ordering of Bhattacharyya parameter is $Z(W_{g,q}^{(i+1)}) \preceq Z(W_{g,q}^{(i)})$. Note that  $W_{G'_{q}}^{(i)}$ can be seen as a degradation of channel $W_{g,q}^{(i)}$ which can be expressed as
\begin{equation}
Z(W_{G'_{q}}^{(i)})= Z(W_{g,q}^{(i)})+\Delta(W,i)
\end{equation}
where $\Delta(W,i)$ is the incremental function of the Bhattacharyya parameter after channel degradation, which is related to the input channel $W$ and the sub-channel index $i$. Note that in the single-layer RS kernel, the degradation decreases with increasing channel index $i$. Therefore, it's easy to express that the relationship for the channel  reliability of the sub-channels in RS kernel  as $W_{G_q}^{(i)} \preceq W_{G_q}^{(i+1)}$.

\subsection{ Partial Orders of RS Polar Codes }
 \newtheorem{Proposition}{Proposition}
\begin{Definition} \label{Definition3-2}
(Addition Operator): Denote $(i_{m-1},...,i_1,i_0)$ as the  $q$-ary representation of the index $i$ in (\ref{e6}). Given $k\in\left\{0,1,...,m-1\right\}$, the Addition Operator at position $k$ maps $i$ into $A^{(k)}(i)\in\left\{0,1,...,N-1\right\}$. If $i_k=q-1$ then $A^{(k)}(i)=i$. Otherwise, the $q$-ary representation of $A^{(k)}(i)$ is defined as 

\begin{equation}
\left(A^{(k)}(i)\right)_{\ell}= \begin{cases}i_{\ell}+1, & \ell=k, \\ i_{\ell}, & \ell \neq k .\end{cases}
\end{equation}
\end{Definition}

\begin{Proposition} \label{Proposition1}
\rm
Let $W$ be a $q$-ary  memoryless output symmetric channel $W$ and the synthetic sub-channel $W_N^{(i)}$ is  obtained from $W$ by applying (\ref{ei}). Then, for any $i \in  \left\{0,1,...,N-1\right\}$, $W_{N}^{(i)} \preceq  W_{N}^{(A^{(k)}(i))}$.
\end{Proposition}
\begin{proof} \label{proof1}
Let $i$ be the index of synthetic sub-channel $W_N^{(i)}$, and let the $q$-ary representation of $i$ be  $(i_{m-1},...,i_1,i_0)$, where $i \in  \left\{0,1,...,N-1\right\}$ and $i_{k}\in\left\{0,1,...,q-1\right\}$. The $A^{(k)}(i)$ operation is the selection of different split channels in a single-layer RS kernel. As shown in the previous subsection, for the same  $q \times q$ RS kernel $G_q$, the reliability representation of  sub-channels can be express as $W_{G_q}^{(i_k)} \preceq W_{G_q}^{(i_k+1)}$.
This means that compared to sub-channel $i$, the reliability of the input channel of sub-channel $A^{(k)}(i)$ increases from the $k$-th layer kernel, while the subsequent polarization process remains unchanged, so the final channel reliability is higher. Thus, for any $i \in  \left\{0,1,...,N-1\right\}$, $W_{N}^{(i)} \preceq  W_{N}^{(A^{(k)}(i))}$ which  concludes the proof.
\end{proof}

\begin{Example} \label{Example1}
\rm
Consider the RS polar code over $\mathcal{GF}(4)$ with  $m=3$ and index $i=25$, its $q$-ary representation is $(1,2,1)$. Then, after $A^{(1)}(25)$, we get the  $i=29$ with $(1,3,1)$. Applying the  Proposition 1, we conclude $W_{64}^{(25)} \preceq  W_{64}^{(29)}$.

\end{Example}

\begin{Definition} \label{Definition3-3}
(Left-Swap Operator): Denote $(i_{m-1},...,i_1,i_0)$ as the  $q$-ary representation of the index $i$ in (\ref{e6}). Given $k_1,k_2\in\left\{0,1,...,m-1\right\}$ and $k_1<k_2$, the Left-Swap operator at position $k_1$ and $k_2$ maps $i$ into $L^{(k_1,k_2)}(i)\in\left\{0,1,...,N-1\right\}$. If $i_{k_1}\leq i_{k_2}$ then $L^{(k_1,k_2)}(i)=i$. Otherwise, the $q$-ary representation of $L^{(k_1,k_2)}(i)$ is defined as 

\begin{equation}
\left(L^{(k_1,k_2)}(i)\right)_{\ell}= \begin{cases}i_{k_2}, & \ell=k_1,\\ i_{k_1}, & \ell = k_2 .\\ i_{\ell}, & \ell  \notin \left\{k_1,k_2\right\} .\end{cases}  \label{e5}
\end{equation}

\end{Definition}

\begin{Proposition} \label{Proposition2}
\rm
Let $W$ be a $q$-ary  memoryless output symmetric channel $W$ and the synthetic sub-channels $W_N^{(i)}$ are  obtained from $W$ by applying (\ref{ei}). Then, for any $i \in  \left\{0,1,...,N-1\right\}$, $W_{N}^{(i)} \preceq  W_{N}^{(L^{(k_1,k_2)}(i))}$.
\end{Proposition}
\begin{proof} \label{proof2}
Let $W$ be a $q$-ary  memoryless output symmetric channel. First consider the case $k_2=k_1+1$,  which can be seen as two  adjacent polarization operations  based on the RS kernel $G_q$ i.e $G_{q}^{\otimes 2}$. Denote $N_2=q^2$, the 
transition probability of suchannel $W_{N_{2}}^{(i)}$ is shown as 

\begin{equation}
W_{N_{2}}^{(i)}\left(y_0^{N_{2}-1}, s_0^{i-1} \mid s_i\right)=\frac{1}{q^{N_{2}-1}} \sum_{s_{i+1}^{N_{2}-1}} W_{N_{2}}^{(i)}\left(y_0^{N_{2}-1} \mid s_0^{N_{2}-1}\right)
\end{equation}

Consider the  channel $W_{g,N_{2}}^{(i)}$ as the genie-added channel of $W_{N_{2}}^{(i)}$ and its transition probability is defined as 

\begin{equation}
\begin{aligned}
W_{g,N_{2}}^{(i)}\left(y_0^{N_{2}-1}, s_0^{i-1}, s_{i+1}^{N_{2}-1}\mid s_i\right)\\
=\frac{1}{q^{N_{2}-1}}  \prod_{j=0}^{N_{2}-1} W\left(y_i \mid\left(s_0^{N_{2}-1} G_{q}^{\otimes 2}\right)_j\right)
\end{aligned}
\end{equation}

According to the proof in the previous subsection, $W_{g,N_{2}}^{(i)}$ is equivalent to $(\tilde{W}^{\circledast k_1})^{\circledast k_2}$ that is
\begin{equation}
Z(W_{g,N_{2}}^{(i)})=Z(W_{g,N_{2}}^{(L^{(k_1,k_2)}(i))})=Z(W)^{(i_{k_1}+1)*(i_{k_2}+1)}
\end{equation}
Here, we can assume that in the ideal case,  $s_{L^{(k_1,k_2)}(i)}$ and $s_i$ pass through channels $W_{g,q}^{(i_{k_1})}$ and $W_{g,q}^{(i_{k_2})}$ in different orders with the similar  degradation pattern, but more outputs are omitted in $ W_{N_{2}}^{(i)} $ compared to $W_{N_{2}}^{(L^{(k_1,k_2)}(i))} $. Therefore, $ W_{N_{2}}^{(i)} $ is stochastically
degraded with respect to $W_{N_{2}}^{(L^{(k_1,k_2)}(i))} $ i.e.

\begin{equation}
\begin{aligned}
W_{N_{2}}^{(L^{(k_1,k_2)}(i))} 
& \equiv s_{i}\rightarrow\left(y_0^{N_{2}-1},s_0^{L^{(k_1,k_2)}(i)-1}\right) \\
& \succeq s_{i} \rightarrow\left(y_0^{N_{2}-1},s_0^{i-1}\right) \\
& \equiv W_{N_{2}}^{(i)}  \label{e7}
\end{aligned}
\end{equation}

 On the basis of the adjacent RS kernel, the reliability of the sub-channel can be passed to the left in the $q$-ary representation index  through the expression (\ref{e7}). Thus, when  $k_2-k_1>1$, $W_{N}^{(L^{(k_1,k_2)}(i))}\succeq W_{N}^{(i)}$ also holds which concludes the proof.

\end{proof}

\begin{Example} \label{Example2}
\rm
Consider the RS polar code based $\mathcal{GF}(4)$ with  $m=3$ and index $i=27$, its $q$-ary representation is $(1,2,3)$. Then, after $L^{(0,2)}(27)$, we get the index $i=57$ with $(3,2,1)$. Applying the  Proposition 2, we conclude $W_{64}^{(27)} \preceq  W_{64}^{(57)}$.

\end{Example}

\begin{Proposition} \label{Proposition3}
(Quasi-Nested )
\rm
 The reliability orders of subchannels determined for an RS polar code of length $N$ by partial orders remain unchanged in code length of $qN$.

\end{Proposition}

\begin{proof} \label{proof3}
Denote  $(i_{m-1},...,i_1,i_0)$ as  the  $q$-ary representation of RS kernel polar code in length $N$. Then, in the length $qN$, $(0,i_{m-1},...i_1,i_0)$ also follows the above two partial orders.

\end{proof}

\subsection{ Partial Distance-based Polarization Weight Construction Method}
The concept of Polarization Weight (PW) was initially introduced in \cite{b1-26} to construct polar codes with Arikan kernel. In the previous subsection, we demonstrated that the reliability of a sub-channel is closely linked to the partial distance of its corresponding index in the single-layer RS kernel. Moreover, holding other conditions constant, the reliability of the subchannel increases with the number of layers. As a result, we propose a Partial Distance-based Polarization Weight (PDPW) construction method in this section, which has low computational complexity.

Our investigation revealed that, in RS polar codes, the reliability of a sub-channel with an index $i$ $(i_{m-1},..,i_1,i_0)$, where $i_k\in \left\{0,1,...,q-1\right\}$, can be seen as being affected by two factors.  We refer to  these factors as  the intra-layer polarization effect  associated with $i_k$ and the inter-layer polarization effect related to $k$. Since the reliability of sub-channel with genie-added $W_{g,N}^{(i)}$ can be represented by $Z(W)^{\prod_{k=0}^{m-1}D_{i_k}}$ with $0<Z(W)<1$, we represent the reliability of sub-channel with genie-added by
\begin{equation}
w_{g,N}(i)\mapsto \prod_{k=0}^{m-1}D_{i_k}
\end{equation}
 where the higher the value of  $w_{g,N}(i)$ is, the greater its reliability. However, since  $W_{g,N}(i)$ is an ideal situation, we factor in the two influencing factors mentioned earlier to adjust its value. Using the concept of polarization weight, we represent the reliability of a sub-channel $W_N^{(i)}$ by
\begin{equation}
w(i)\mapsto \prod_{k=0}^{m-1}D_{i_k}^{\zeta(i_k)\beta^k}
\end{equation}
where $\zeta(i_k)$ and $\beta^k$ denote the correction factors for the intra-layer polarization effect and inter-layer polarization effect, respectively. To simplify the computation, we introduce logarithms and  derive the final form of the PDPW construction method as shown in Definition \ref{Definition3-4}.

\begin{Definition} \label{Definition3-4}
(PDPW):  For a $N=q^m$ polar code with RS kernel, denote $(i_{m-1},...,i_1,i_0)$ as the  $q$-ary representation of the index $i$ in (\ref{e6}) and $D_{i_k}$ as the the partial distance of the index $i_k$ in a single-layer kernel. Then, the Partial Distance-based Polarization Weight $w(i)$ is defined as 
\begin{equation}
f^{\mathrm{PDPW}}: w(i) \mapsto \sum_{k=0}^{m-1} \zeta(i_k)\beta^k log_2D_{i_k} 
\end{equation}

where $\beta^k$ and  $\zeta(i_k)$ are the  correction factors for inter-layer polarization effect  and  intra-layer polarization effect, respectively. These numbers require careful selection, backed by comprehensive experimentation and justification.

\end{Definition}

As discussed earlier,  the  sub-channel  $W_{G_q}^{(i)}$ in the single-layer RS kernel can be interpreted  as  the degradation of the genie-added channel $W_{g,q}^{(i)}$. Consequently, $\zeta(i)$ can be viewed as the proportion of the  sub-channel's reliability relative to the genie-added channel in a single RS kernel. Through the simulation tests, we observed  that under certain channel conditions,  $\zeta(i)$ can be approximated by mutual information as below:


 \begin{equation}
\zeta(i)=\frac{I(W_{G_q}^{(i)})}{I(W_{g,q}^{(i)})}
\end{equation}

Referring to \cite{b1-28}, one can use Monte-Carlo simulations to approximate the mutual information of the corresponding channel by calculating the entropy of the symbol probabilities in the single-layer RS kernel and averaging them over sufficiently large number of channel output instances. Denote $H(S_i)$ as the entropy function of $S_i$, $Pr_i(\eta)=W_{G_{q}}^{(i)}\left(Y_0^{q-1}, S_0^{i-1} \mid S_i=\eta\right)$ and $Pr'_i(\eta)=W_{g,q}^{(i)}\left(Y_0^{q-1}, S_0^{i-1},S_{i+1}^{q-1} \mid S_i=\eta\right)$
, where $S_i$ is a random variable corresponding to the $i$-th input symbol of the single RS kernel and $\eta\in\mathcal{GF}(q)$. Then the corresponding mutual information functions can be  calculated as follows:

\begin{equation}
\begin{aligned}
I(W_{G_q}^{(i)}) & =H(S_i)-H(Y_0^{q-1},S_0^{i-1}|S_i) \\
&\approx \frac{1}{T}\sum_{z=0}^{T} \hat{I}_z(W_{G_q}^{(i)})
\end{aligned}
\end{equation}
 where 
\begin{equation}
 \hat{I}_z(W_{G_q}^{(i)})=H(S_i)+\sum_{\eta\in\mathcal{GF}(q)} Pr_i(\eta)log_2(Pr_i(\eta))
\end{equation}

Since $W_{g,q}^{(i)}$ is the channel with genie-added information, we assume that the Bhattacharyya parameter of $W_{g,q}^{(i)}$ in the single-layer RS kernel reaches its  average error probability bound i.e. 
\begin{equation}
P_e(W_{g,q}^{(i)})\approx (q-1)Z(W_{g,q}^{(i)})=(q-1)Z(W)^{D_i}  
\end{equation}
  We can then approximate the value of $Pr'_i(\eta)$ by  $W_{G_q}^{(q-1)}=W_{g,q}^{(q-1)}$. Next, calculating $I(W_{g,q}^{(i)})$  by using  the above mutual information approximation functions and finally obtaining the $\zeta(i)$ corresponding to the sub-channel within a single-layer RS kernel. Note that the coefficient $\zeta(0)$ can be directly assigned a value of zero as it does not make a  contribution to the polynomial. Similarly, the coefficient  $\zeta(q-1)$ can be directly set to one as the entire assumption is based on the premise of $W_{G_q}^{(q-1)}=W_{g,q}^{(q-1)}$.

 For the coefficient $\beta$, as $k$ is the non-negative, $\beta^k>0$ which follows the Addition partial order. To achieve the Left-Swap partial order, $\beta$ should  satisfy $\beta^k<\beta^{k+j}$ for any non-negative integer $k$ and positive integer $j$, which implies that $\beta>1$. Once the coefficient $\zeta(i_k)\log_2D_{i_k}$ is obtained, the value of $\beta$ can be approximated using a polynomial equation solver, as described in \cite{b1-26}.

\begin{Example} \label{Example3}
\rm
Taking RS kernel $G_4$ as an example. The coefficient $\zeta(i)$ of $G_4$ at different signal-to-noise ratios (SNR) over additive white Gaussian noise (AWGN) channels are shown in  Table I. Here we choose the $\zeta(i)$ values at Eb/N0=-1.8dB as in the Monte-Carlo method, which have the best pre-test performance.

 \begin{figure}[!h]
\centering
\includegraphics[width=3.2in]{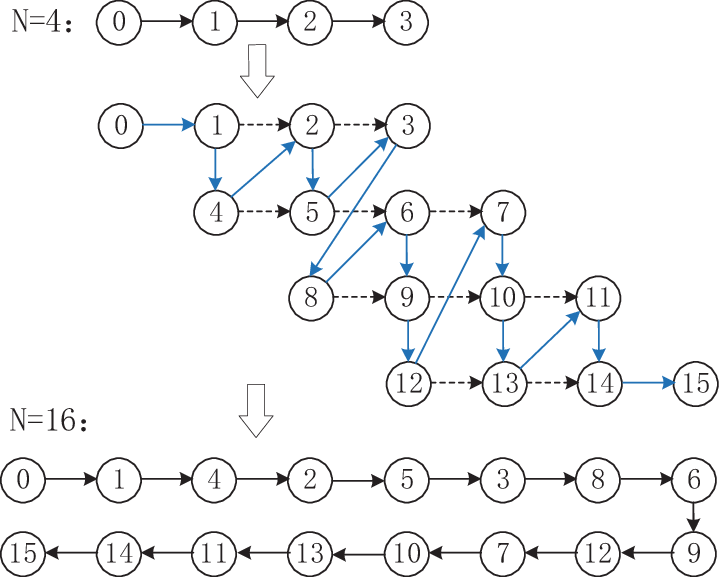}
\caption{ Polarization weights of sub-channels with block length N=256 }
\label{fig:figure3}
\end{figure}
When the block length increases, the uncertainty relation among them is calculated according to the actual situation, and the interval of  $\beta$ is approximated. Fig.3. illustrates how the range of $\beta$ evolves as the block length increases from N=4 to N=16 over AWGN channels. The uncertainty relations including (3,8), (7,12), (10,13) etc. are calculated on the basis of $\beta>1$ as follows:

\begin{equation}
\begin{aligned}
            7 \succ 12     \iff  \beta <1.55         \nonumber  \\
            13 \succ 10     \iff  \beta >1.12      \nonumber  \\
              8 \succ 3     \iff  \beta >1.437      \nonumber  \\
\end{aligned}
\end{equation}
Thus, we obtain an approximate range $(1.437, 1.55)$ for $\beta$. Then, by increasing the block size, a more accurate range of  $\beta$ can be calculated.

Let $\beta=1.512$  and the  block length is set to  $N=256$. Then, the weight of the sub-channel with index $99$, i.e., $(1,2,0,3)$ can be computed as 

\begin{equation}
\begin{aligned}
w(99) &= \zeta(1)\beta^{3}log_{2}2+\zeta(2)\beta^{2}log_{2}3+0+\zeta(3)\beta^{0}log_{2}4 \nonumber\\
      &=7.648
\end{aligned}
\end{equation}
Similarly, we can construct the polarization weights of the sub-channels of the overall codeword, as shown in Fig.\ref{fig:figure4}.
\end{Example}
\begin{figure}[!h]
\centering
\includegraphics[width=3.2in]{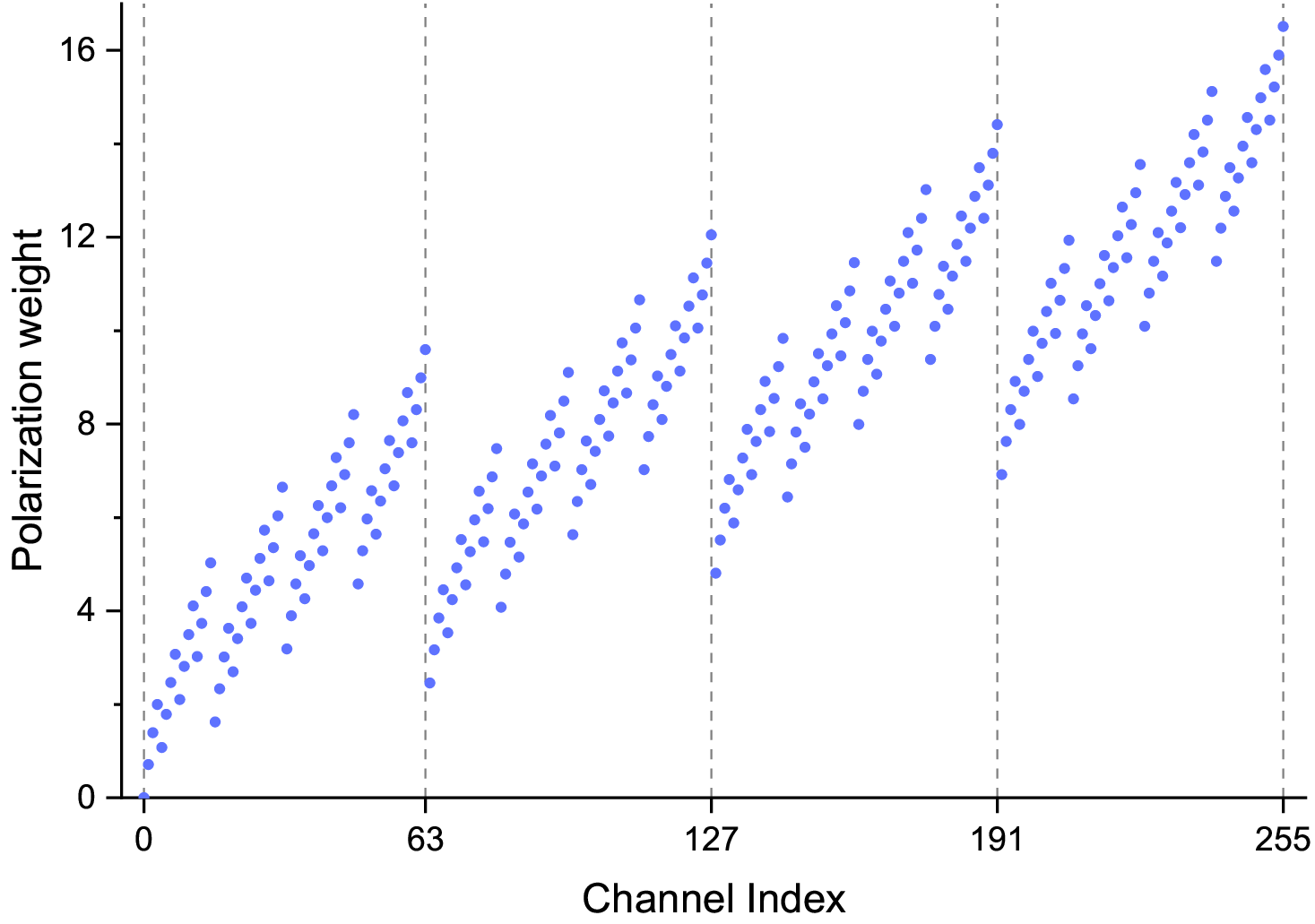}
\caption{ Polarization weights of sub-channels over $G_4$ with block length $N=256$ }
\label{fig:figure4}
\end{figure}

Finally, we analyze the complexity of the proposed construction method. Denote $ \Omega$ as the complexity of a single evaluation of (\ref{ei}). The complexity of the approximate functions $\zeta(i)$ is $\mathcal{O}(Tq\Omega)$ and the PDPW construction complexity is $\mathcal{O}(N)$.  In comparison,  the  complexity of the Monte-Carlo method can be expressed as $\mathcal{O}(TNlog_q(N)\Omega)$. Compared to  classical Monte-Carlo simulation, the proposed PDPW method requires only one approximation of the mutual information of sub-channels for a single-layer RS kernel, thus greatly reducing the construction complexity.

\section{ Minimum Polarization weight Puncturing scheme }\label{sec:section4}

Puncturing  is an efficient rate-matching technique that is commonly used in 5G communication standards for low-rate polar codes \cite{b1-12}.  Fixed  information set is an essential part of incremental redundancy hybrid automatic repeat request (IR-HARQ) in 3GPP NR. Referring to the PDPW construction method  in the previous section,  we propose  Minimum Polarization weight Puncturing (MPWP) algorithm for  fixed  information set which is shown as follows:

Denote the $N$ symbols length  polar codes with RS kernel in bits as $N_b=tq^m$. According to the transmission requirements, the actual number of transmitted  bits is $M_b$, where $M_b<N_b$.

\begin{itemize}
\item Based on the polarization weights of the subchannels, the fixed information set can be selected as $\mathcal{I}$.  Then, from the set $\mathcal{I}^c$,  $l+1$ subchannels with the smallest weights are selected to determine the  puncturing symbols vector $\mathcal{R}_p=(r_{0}, r_{1}, \ldots, r_{l-1}, )$,  where    $l=\lceil (N_b-M_b)/t\rceil$ and $w(r_{0})<w(r_{1})<\ldots< w(r_{l-1})$.
\item On the encoder side, do not transmit the first $l-1$ symbols of $\mathcal{R}_p$ that is freezing them in input symbols $s_{0}^{N-1}$. Note that for index $r_{l-1}$, special judgments are required: Denote $\sigma=(N_b-M_b)$ mod $t$. If  $\sigma>0$, disable the last $\sigma$ bits of $q$-ary representation
of index $r_{l-1}$. Otherwise, disable the  whole symbol $r_{l-1}$.
\item On the decoding side, the decoding architecture is not changed. Since the symbols and bits represented in $\mathcal{R}_p$ are not actually transmitted, their corresponding LLRs are set to 0.
\end{itemize}

\section{SIMULATION RESULTS }\label{sec:section5}
In this section,  we present the simulation results for the proposed construction method and rate-matching algorithm. For the following simulations, the RS polar codes are over $\mathcal{GF}(4)$ and  exploited the CA-SCL decoding algorithm with CRC=8 bits and list size L=2. All codewords are modulated using the binary phase-shift keying (BPSK)  and transmitted over additive white Gaussian noise (AWGN) channels.

In the first trial, we evaluate the performance of the proposed construction method at code rates  R=1/3, R=1/2, and R=2/3. Fig.\ref{fig:figure5} and Fig.\ref{fig:figure6} illustrate the performance of the proposed  PDPW construction method and the Monte-Carlo method with symbol lengths of $N=1024$ and $N=256$, respectively. As expected, the performance of  the two construction methods is essentially the same for the two block lengths and three code rates. Fig.\ref{fig:figure6} also includes the  performance of 3GPP NR polar codes over Arikan kernel with a bit length of  $N_b=512$ for comparison. Overall, the performance of the proposed PDPW construction method over $\mathcal{GF}(4)$ RS kernel  is better than 3GPP polar codes over  Arikan kernel. Specifically when BLER=$10^{-3}$, the proposed method yields nearly 0.23db and 0.27db gain over 3GPP NR polar codes for rates of R=1/3 and R=1/2 respectively.
\begin{figure}[!h]
\centering
\includegraphics[width=3.2in]{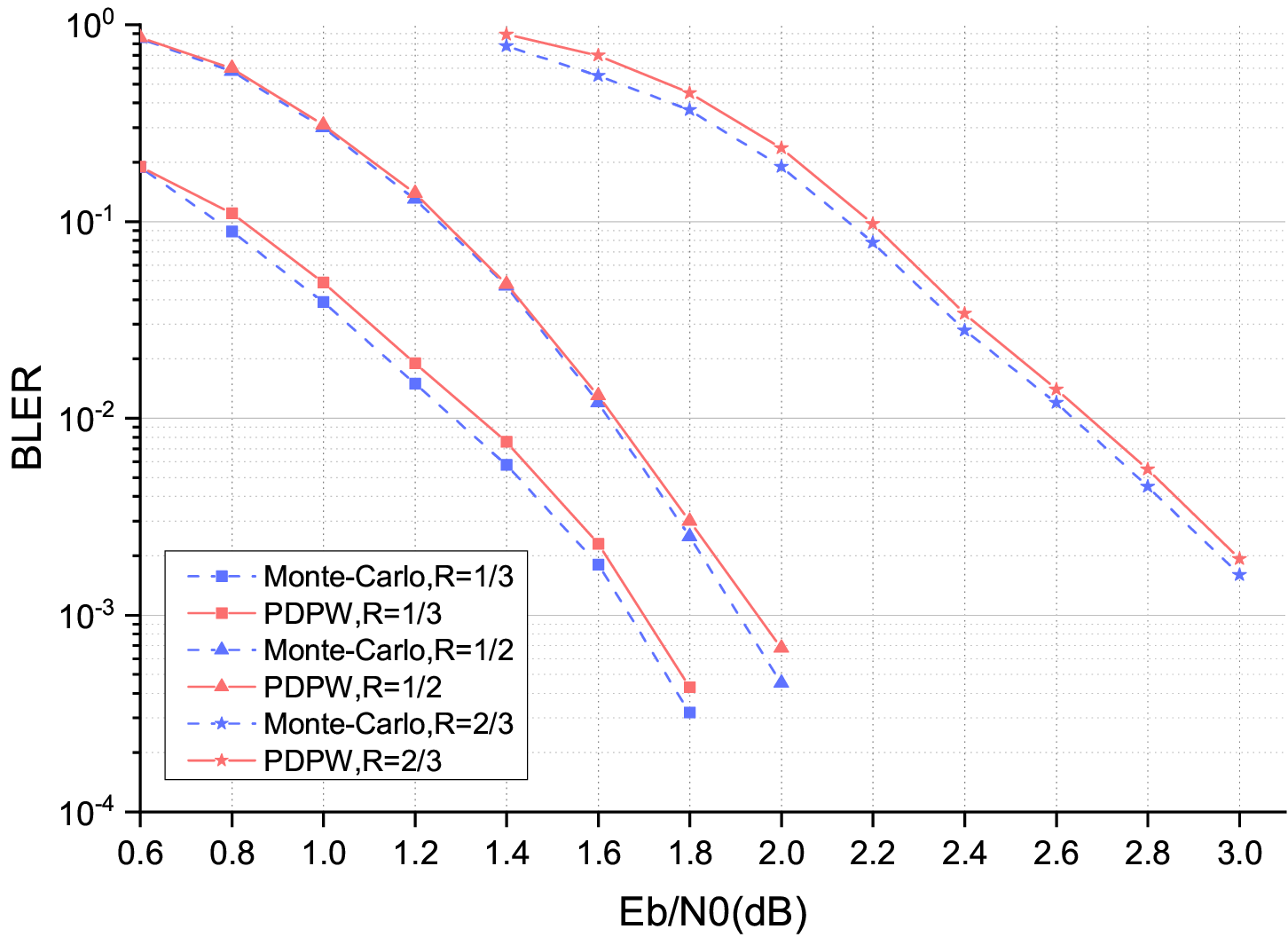}
\caption{ Performance of RS polar codes with block length $N=1024$ }
\label{fig:figure5}
\end{figure}

\begin{figure}[!h]
\centering
\includegraphics[width=3.2in]{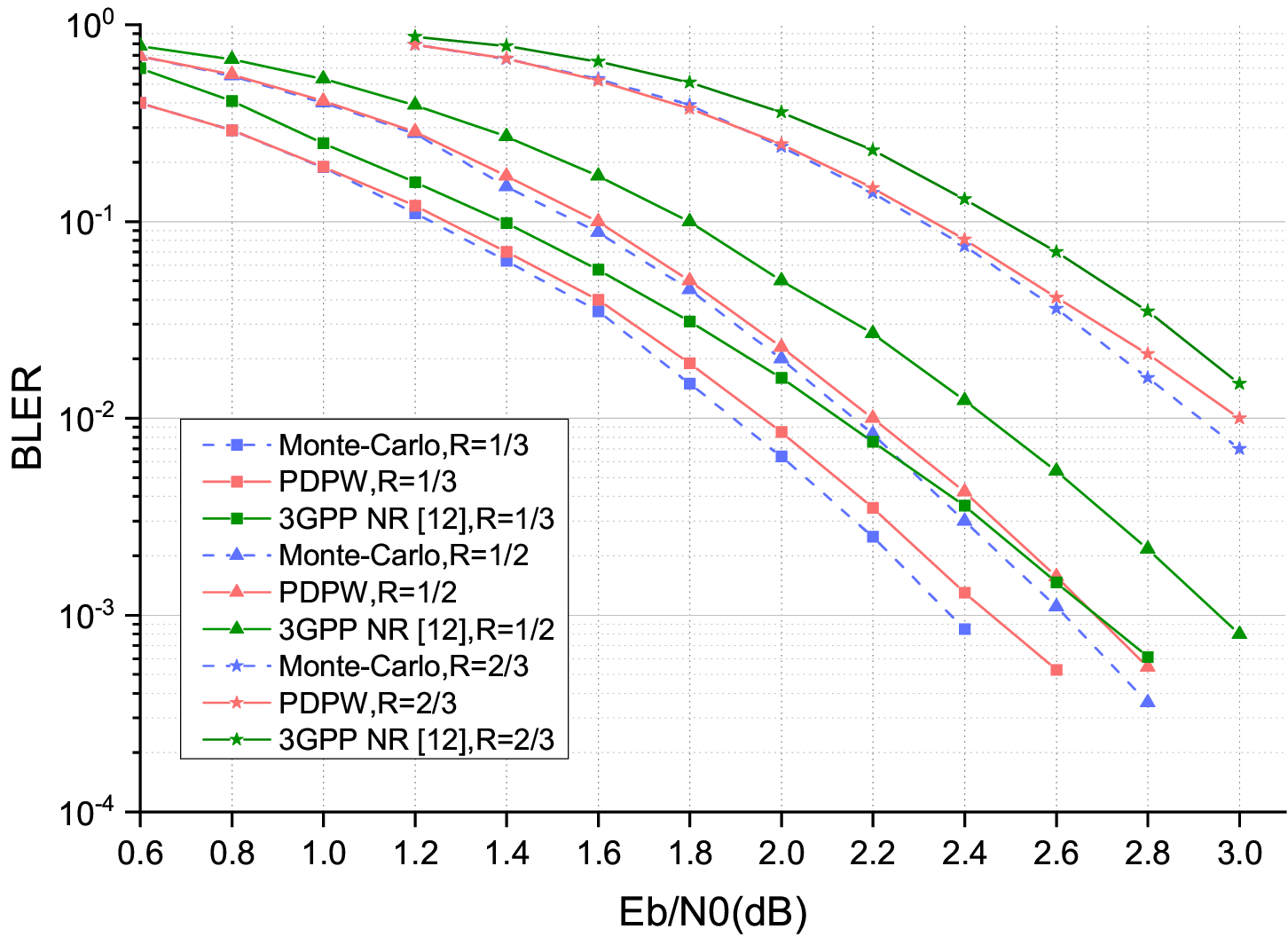}
\caption{ Performance of RS polar codes and 3GPP NR polar codes  with block length $N_b$=512 }
\label{fig:figure6}
\end{figure}

Then, we will mainly compare the rate-matching performance for fixed information set. In the following simulations, the information set is optimized by the mother codes.
Fig.\ref{fig:figure7}. shows the performance of the proposed algorithm and the smallest index puncturing (SIP) algorithm \cite{b1-22} over $\mathcal{GF}(4)$ with mother code length $N=1024$ and actual block length $M=800$. In this trial, our proposed algorithm performs the same as SIP at a low code rate but has a larger gain at a higher code rate. This is because our algorithm retains more highly reliable sub-channels when the code rate is increased.

\begin{figure}[!h]
\centering
\includegraphics[width=3.2in]{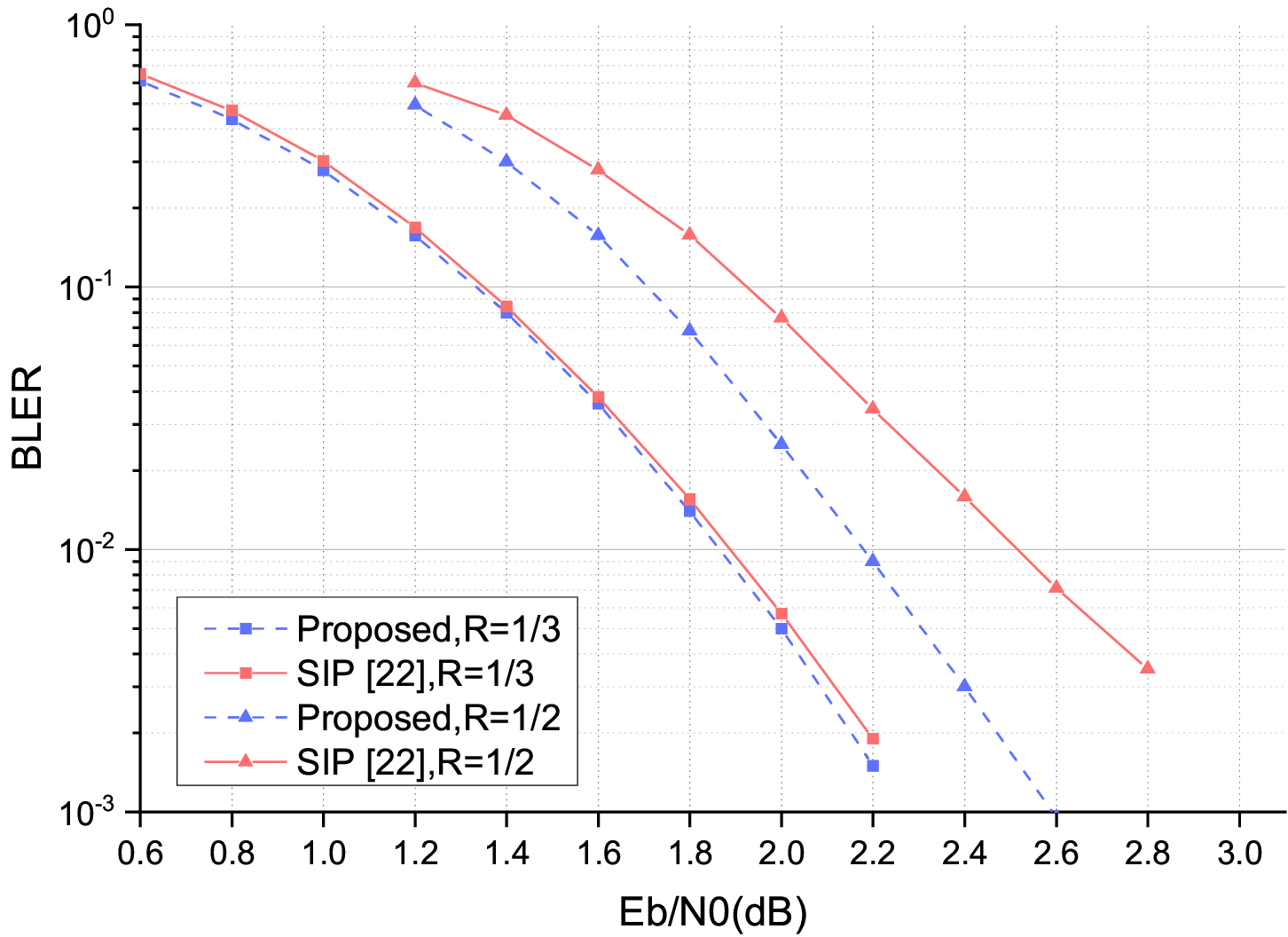}
\caption{  Rate-matching performance of RS polar codes with block length $N=1024$ and $M=800$ }
\label{fig:figure7}
\end{figure}

Furthermore, Fig.8. illustrates the puncturing performance of the proposed method  over $\mathcal{GF}(4)$ RS kernel and the 3GPP NR polar codes  over  Arikan kernel. Note the 3GPP NR polar codes  are constructed utilizing the 3GPP channel reliability table and applied pre-frozen set  regarded as an optimization of the puncturing method in 3GPP NR \cite{b1-12}. This simulation tests the SNR required to achieve BLER=$10^{-3}$ for two codewords at different lengths. The length of the mother code $N_b=512$, the overall code rate is $R=1/3$.  Results indicate that for most lengths greater than 340, the proposed puncturing algorithm for RS polar codes achieves a gain of approximately 0.2dB compared to the 3GPP NR polar codes.

\begin{figure}[!h]
\centering
\includegraphics[width=3.2in]{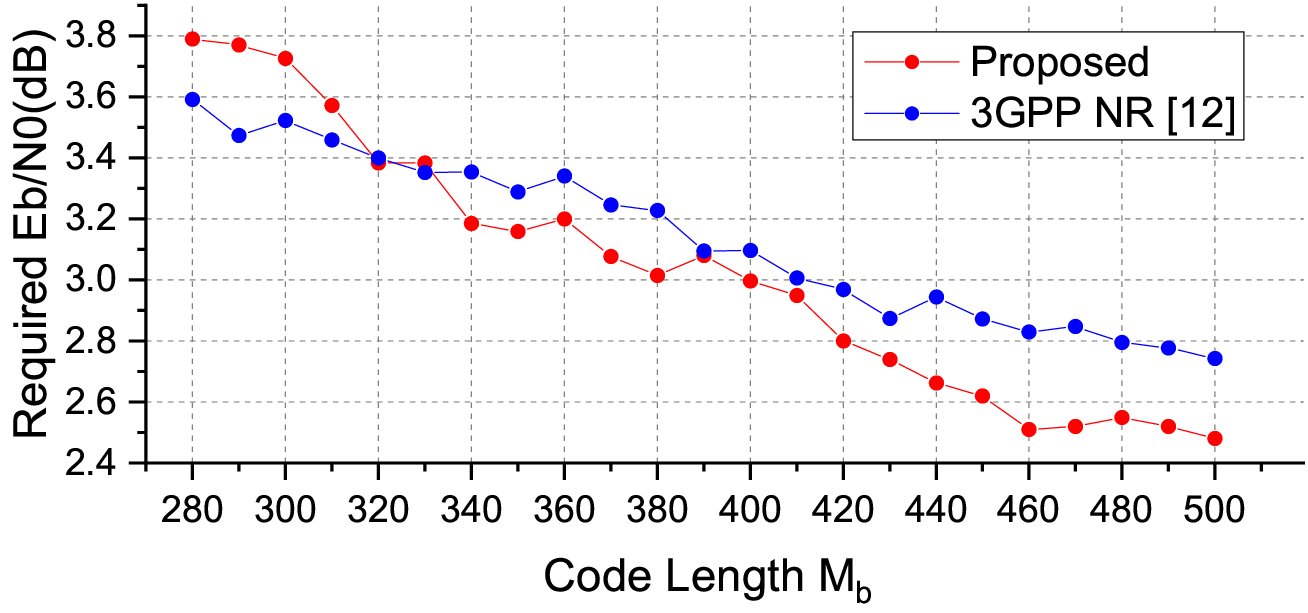}
\caption{ Schematic diagram of the required Eb/N0 for different codes to manage a performance of BLER=$10^{-3}$}
\label{fig:figure8}
\end{figure}

\section{ CONCLUSION }\label{sec:section6}
In this work, we investigate the  partial orders of polar codes with RS kernel and propose a PDPW construction method based on these partial orders, as well as an MPWP puncturing algorithm. The performance of the proposed construction method is nearly identical to that of the Monte Carlo  method, and the proposed  puncturing algorithm outperforms the current SIP algorithm. Based on the proposed construction scheme, both the standard codeword performance and rate-matching performance of RS polar codes are superior to the 3GPP NR polar codes. Therefore, RS polar codes can be considered an important candidate for channel coding in the next generation of communication systems. 

For multi-kernel polar codes, due to the increasing complexity of the kernel, the decoding algorithm becomes a critical area for future research. Specifically, developing a lower-complexity approach for processing soft information in different multi-kernels during the decoding algorithm is of utmost importance for improving the practicality and effectiveness of these codes.

\ifCLASSOPTIONcaptionsoff
  \newpage
\fi

\balance

\end{document}